%% file: ConfPaper.tex
\documentclass[12pt]{amsart}

\newtheorem{algorithm}{Algorithm}
\newtheorem{definition}{Definition}
\newtheorem{theorem}{Theorem}
\usepackage{url,upref,amscd,amsmath}
\usepackage{times}
\usepackage{fancyhdr}
\usepackage[psamsfonts]{amssymb}
\usepackage{cite}
\begin{document}
\title{The ElGamal cryptosystem over circulant matrices}
\author{Ayan Mahalanobis}
\address{Indian Institute of Science
  Education and Research Pune, Pashan Pune-411021, India}
\email{ayan.mahalanobis@gmail.com}
\thanks{Research supported by a NBHM research grant}
\begin{abstract}
In this paper we study extensively the discrete logarithm problem in
the group of non-singular circulant matrices. The emphasis of this
study was to find the exact parameters for the group of circulant
matrices for a secure implementation. We tabulate these parameters. We
also compare the discrete logarithm problem in the group of circulant
matrices with the discrete logarithm problem in finite fields and with
the discrete logarithm problem in the group of rational points of an
elliptic curve. 
\end{abstract}
\maketitle
\section{Introduction}
Two of the most popular groups used in \emph{the discrete logarithm
  problem} are the group of units of a finite field and the group of
rational points of an elliptic curve over a finite field. The obvious question arises, are
there any other groups? I write this paper to show, that there
are matrix groups -- the \emph{group of non-singular circulant matrices},
which is much better than the finite fields in every aspect and
even better than the elliptic curves when one considers the size of the
field for a secure implementation. The size of the field for a secure
implementation is a huge issue in \emph{public key cryptography}. One of the
reasons, elliptic curves are preferred over a finite field discrete
logarithm problem, is the size of the field for a secure
implementation. In our current state of knowledge, it is believed
that the discrete logarithm problem over $\mathbb{F}_{2^{1028}}$ offers
the same security that of most elliptic curves over
$\mathbb{F}_{2^{160}}$.  As our processors get faster and with the
advent of distributed computing these sizes will grow bigger with time.
In the case of an elliptic curve the rate of growth is much smaller than
that of finite fields. We will see, for circulant matrices
the size of the field for a secure implementation can get even
smaller. The comparison of speed, between circulants and elliptic
curves, in an actual implementation is yet to be done. But, since the
circulants use smaller field, it is
likely that the circulants are faster.

It is known~\cite{ayan,silverman2} that the group of \emph{circulant matrices} offers the
same security of a finite field of about same size, with \textbf{half
  the computational cost}. The other interesting fact about circulant
matrices is the \textbf{size of the field} for a secure
implementation. The arithmetic of the circulant matrices is
implemented over a finite field, very similar to the case of elliptic
curves, where the arithmetic is also  implemented over a finite
field. In the case of circulants, the size of the field can be smaller
than the one used for elliptic curves. This is extensively studied in Section 5, and the results are tabulated in Table~\ref{table2}. To sum it up, the advantage of circulants is that it \textbf{uses smaller field and is faster}.

In this paper, we denote the group of non-singular circulant matrices of size $d$ by C$(d,q)$ and the group of special circulant matrices, i.e., circulant matrices with determinant 1, by SC$(d,q)$ respectively.

\begin{definition}[Circulant matrix C$(d,q)$]
A $d\times d$ matrix over a field $F$ is called a circulant matrix, if every row except
the first row, is a right circular shift of the row above that. So a
circulant matrix is defined by its first row.  One can define a circulant 
matrix similarly using columns. 
\end{definition}
A matrix is a two dimensional object, but a circulant matrix behaves
like a one dimensional object -- given by the
first row or the first column. We will denote
a circulant matrix $C$ of size $d$, with the first row $c_0,c_1,\ldots,c_{d-1}$, by
$C=\text{circ}\left(c_0,c_1,c_2,\ldots,c_{d-1}\right)$.
An example of a circulant $5\times 5$ matrix is:
\[
\begin{pmatrix}
c_0&c_1&c_2&c_3&c_4\\
c_4&c_0&c_1&c_2&c_3\\
c_3&c_4&c_0&c_1&c_2\\
c_2&c_3&c_4&c_0&c_1\\
c_1&c_2&c_3&c_4&c_0
\end{pmatrix}
\]
 One can define a
\emph{representer polynomial} corresponding to the circulant matrix
$C$ as $\phi_C=c_0+c_1x+c_2x^2+\ldots+c_{d-1}x^{d-1}$. The circulants form a commutative ring under matrix
multiplication and matrix addition and is isomorphic to (the
isomorphism being circulant matrix to the representer polynomial)
$\mathcal{R}=\dfrac{F[x]}{x^d-1}$. For more on circulant matrices, see
\cite{davis}.

We will study the discrete logarithm problem in SC$(d,q)$, the special circulant matrix. It is fairly straightforward to see that one can develop a Diffie-Hellman key exchange protocol or the ElGamal cryptosystem from this discrete logarithm problem. The ElGamal cryptosystem over SL$(d,q)$, \emph{the special linear group} of size $d$ over $\mathbb{F}_q$ is described below. Since the special circulant matrix is contained in the special linear group, this description of the ElGamal cryptosystem works for SC$(d,q)$ as well.

All \textbf{fields considered in this paper} are finite and of characteristic $2$.
\section{The ElGamal over SL$(d,q)$}
\begin{description}\label{keyex}
\item[Private Key] $m$, $m\in\mathbb{N}$.
\item[Public Key] $A$ and $A^m$. Where $A\in\text{SL}(d,q)$.
\end{description}
\paragraph{\textbf{Encryption}}
\begin{description}
\item[a] To send a message (plaintext) $\mathbf{v}\in\mathbb{F}_q^d$, Bob computes $A^r$
  and $A^{mr}$ for an arbitrary $r\in\mathbb{N}$.
\item[b] The ciphertext is $\left(A^r,A^{mr}\mathbf{v}^\text{T}\right)$. Where $\mathbf{v}^\text{T}$ is the transpose of $\mathbf{v}$.
\end{description}
\paragraph{\textbf{Decryption}}
\begin{description}
\item[a] Alice knows $m$, when she receives the ciphertext
 $\left(A^r,A^{mr}\mathbf{v}^\text{T}\right)$, she computes $A^{mr}$ from $A^r$,
  then $A^{-mr}$ and then computes $\mathbf{v}$ from $A^{mr}\mathbf{v}^\text{T}$.
\end{description}
 
We show that the security of the ElGamal cryptosystem over SL$(d,q)$, 
is equivalent to the Diffie-Hellman problem in SL$(d,q)$. Since SC$(d,q)$ is contained in SL$(d,q)$, this proves that the security of ElGamal cryptosystem is equivalent to the Diffie-Hellman problem in SC$(d,q)$.
 
Assume that Eve can solve
the Diffie-Hellman problem, then from the public information, she
knows $A^m$. From a ciphertext $\left(A^r,A^{rm}\mathbf{v}^T\right)$ she gets
$A^r$. Since she can solve the Diffie-Hellman problem, she computes
$A^{rm}$ and can decrypt the ciphertext. The converse follows
from the following theorem, which is an adaptation of~\cite[Proposition 2.10]{silverman}

\begin{theorem}
Suppose Eve has access to an oracle that can decrypt arbitrary
ciphertext of the above cryptosystem for any private key, then she can solve the
Diffie-Hellman problem in SL$(d,q)$.
\end{theorem}
\begin{proof}
Let $g=A^a$ and $h=A^b$. Eve takes an arbitrary element
$\mathbf{v}$ in the vector space of dimension $d$ on which SL$(d,q)$
acts. We use the same basis used for the representation of SL$(d,q)$. 
Then $\mathbf{v}=(\mathbf{v}_1,\mathbf{v}_2,\ldots,\mathbf{v}_d)$ where
$\mathbf{v}_i\in\mathbb{F}_q^\times$. Let $\widehat{\mathbf{v}}_i=(0,\ldots,\mathbf{v}_i,\ldots,0)$
and $c=\widehat{\mathbf{v}}_i^T$. She pretends that $A$ and $A^a$ is a public
key. Sends that information to the oracle. Then asks the oracle to
decrypt $(h,c)$. Oracle sends back to Eve, $h^{-a}c$. Eve knowing
$\mathbf{v}$, computes the $i\textsuperscript{th}$ column of
$A^{-ab}$ from $h^{-a}c$. In $d$ tries $A^{ab}$ is found.
This solves the Diffie-Hellman problem.
\end{proof}

\section{Security of the proposed ElGamal cryptosystem}\label{sec1}
This paper is primarily focused on the discrete logarithm
problem in the automorphism group of a vector space over a finite field. There are
two kinds of attack on the discrete logarithm problem.
\begin{itemize}
\item[(i)] The ``so called'' generic attacks, like
the \emph{Pollard's rho} algorithm. These attacks use a \emph{black box} group
algorithm. The time complexity of these algorithms is about the same as the
square-root of the size of the group.
\item[(ii)] The other one is an \emph{index calculus} attack. These attacks do not work in any group.
\end{itemize}
Black box group algorithms work in any group, hence they will
work in SC$(d,q)$ as well. The most efficient way to use black box attack on the discrete logarithm problem, is to use the Pohlig-Hellman algorithm~\cite[Section 2.9]{silverman} first. This reduces the discrete logarithm problem to the prime divisors of the order of the element (the base for the discrete logarithm) and then use the Chinese remainder theorem to construct a solution for the original discrete logarithm problem. One can use the Pollard's rho algorithm to solve the discrete logarithm problem in the prime divisors. So the whole process can be summarized as follows: the security of the discrete logarithm against generic attacks, is the security of the discrete logarithm in the largest prime divisor of the order. 
We cannot prevent these attacks. These generic attacks are of exponential time complexity and are not of
much concern.

The biggest threat to any cryptosystem using the
discrete logarithm problem is a subexponential attack like the index
calculus attack~\cite{oliver}. It is often argued~\cite{koblitz,joseph}
that there is no 
index calculus algorithm for most elliptic curve cryptosystems
that has subexponential time complexity. This fact is often
used to promote elliptic curve cryptosystem over a finite
field cryptosystem~\cite{koblitz}. So, the best we can hope from the
discrete logarithm problem in SC$(d,q)$ is, there is no index calculus attack or the index calculus attack
becomes exponential.

The expected asymptotic complexity of the index calculus
algorithm in $\mathbb{F}_{q^k}$ is 
$\exp{\left((c+o(1))(\log{q}^k)^\frac{1}{3}(\log\log{q}^k)^\frac{2}{3}\right)}$
, where $c$ is a constant, see~\cite{oliver} and~\cite[Section 4]{koblitz}.
If the degree of the extension, $k$, is greater than
$\log^2{q}$ then the asymptotic time complexity of the index calculus
algorithm becomes exponential. In our 
case this means, if $d>\log^2{q}$, the asymptotic complexity
of the index calculus algorithm on circulant matrices of size $d$ becomes exponential.

If we choose $d\geq\log^2{q}$, then the discrete logarithm problem in SC$(d,q)$ becomes as
secure as the ElGamal over an elliptic curve, because the index
calculus algorithm is exponential; otherwise we can not
guarantee. But on the other hand, in the proposed cryptosystem,
encryption and decryption 
works in $\mathbb{F}_q$ and breaking the cryptosystem depends on
solving a discrete logarithm problem in $\mathbb{F}_{q^{d-1}}$. Since,
implementing the index calculus attack becomes harder as the field
gets bigger. It is clear that if we take 
$d\ll\log^2{q}$, then the cryptosystem is much more secure than the
ElGamal cryptosystem over $\mathbb{F}_q$.

\section{Is the ElGamal cryptosystem over SC$(d,q)$ really useful?}\label{sec2}
For a circulant matrix over a field of even characteristic, squaring is fast. It is shown~\cite[Theorem
2.2]{ayan} that, if
$A=\text{circ}\left(a_0,a_1,\ldots,a_{d-1}\right)$, then
$A^2=\text{circ}\left(a^2_{\pi(0)},a^2_{\pi(1)},\ldots,a^2_{\pi(d-1)}\right)$. Where
$\pi$ is a permutation of $\{0,1,2,\ldots,d-1\}$. Now the $a_i$s
belong to the underlying field $\mathbb{F}_q$ of characteristic $2$.
In this field, squaring
is just a cyclic shift using a \emph{normal basis}~\cite[Chapter 4]{mullin} representation of the
field elements.

It was shown by Mahalanobis~\cite{ayan}, that if five
conditions are satisfied, then the security of the discrete logarithm
problem for circulant matrices of size $d$ over $\mathbb{F}_q$ is the
same as the discrete logarithm problem in $\mathbb{F}_{q^{d-1}}$.

The five conditions are:
\begin{itemize}
\item[a.] The circulant matrix should have determinant 1.
\item[b.] The matrix $A$ should have row-sum 1.
\item[c.] The integer $d$ is prime.
\item[d.] The polynomial $\dfrac{\chi_A}{x-1}$ is irreducible.
\item[e.] $q$ is primitive mod $d$.
\end{itemize}
In short, the argument for these five conditions are the following:

Let $A=\text{circ}\left(a_0,a_1,\dots,a_{d-1}\right)$ and let $\chi_A$
be the characteristic polynomial of $A$. It is easy to see that the row-sum, $a_0+a_1+\cdots+a_{d-1}$,
sum of all elements in a row, is constant for a circulant matrix. This
row-sum, $\alpha$ is an eigenvalue of $A$ and belongs to
$\mathbb{F}_q$. Clearly, $\alpha^m$ is an eigenvalue of $A^m$.
This $\alpha$ and $\alpha^m$ can
reduce a part of the discrete logarithm problem in $A$, to a discrete logarithm
problem in the field $\mathbb{F}_q$. If the row-sum is 1, then there
is no such issue. This is the reason behind the condition, the row-sum
is 1.

Now assume that $\dfrac{\chi_A}{x-1}=f_1^{e_1}f_2^{e_2}\ldots
f_n^{e_n}$, where each $f_i$ is an irreducible polynomial and $e_i$s  are positive
integers\footnote{Condition c.~ensures that $e_i=1$ for all $i$.}. Then it follows,
the discrete logarithm problem in $A$, can be reduced to discrete
logarithm problems in $\dfrac{\mathbb{F}_q[x]}{f_i}$, for each $i$. Then
one can solve the individual discrete logarithms in extensions
of $\mathbb{F}_q$, put those solutions together using the Chinese remainder theorem and solve the discrete logarithm problem in $A$. The degree of these extensions, the size of which provides us
with the better security, is maximized when $\dfrac{\chi_A}{x-1}$ is
irreducible. This is the reason for $\dfrac{\chi_A}{x-1}$ is irreducible.

The ring of circulant matrices is isomorphic to $\dfrac{\mathbb{F}_q[x]}{x^d-1}$,
moreover $\dfrac{\mathbb{F}_q[x]}{x^d-1}$ is isomorphic to
$\dfrac{\mathbb{F}_q[x]}{x-1}\times\dfrac{\mathbb{F}_q[x]}{\Phi(x)}$,
where $\Phi(x)=\dfrac{x^d-1}{x-1}$ is the $d\textsuperscript{th}$ cyclotomic polynomial. If $d$ is prime and
$q$ is primitive modulo $d$, then the cyclotomic polynomial $\Phi(x)$
is irreducible. In this case, the discrete logarithm problem in
circulant matrices reduce to the discrete logarithm problem in
$\mathbb{F}_{q^{d-1}}$. 
\subsection{What are the advantages of using circulant matrices?}\label{sec21}
The advantages of using circulant matrices are:
\begin{itemize}
\item Multiplying circulant matrices of size $d$ over $\mathbb{F}_q$ is twice as fast compared to
  multiplication in the field of size $\mathbb{F}_{q^d}$.
\item Computing
the inverse of a circulant matrix is easy.
\end{itemize}
 Since any circulant
matrix $A$ can be represented as a
polynomial of the form $f(x)=c_0+c_1x+\ldots+c_{d-1}x^{d-1}$. This
polynomial is invertible, implies that, 
$\gcd\left(f(x),x^d-1\right)=1$. Then one can use the extended
Euclid's algorithm to find the inverse. In our cryptosystem, we need to
find that inverse, and it is easily computable.

We now compare the following three cryptosystems for security and
speed. We do not
compare the key sizes and the size of the ciphertext, as these can be
decided easily.
\begin{itemize}
\item[1.] The ElGamal cryptosystem using the circulant matrices of
  size $d$ over $\mathbb{F}_q$.
\item[2.] The ElGamal cryptosystem using the group of an elliptic
  curve.
\item[3.] The ElGamal cryptosystem over $\mathbb{F}_{q^d}$.
\end{itemize}
\subsection{ElGamal over $\mathbb{F}_{q^d}$ vs.~the circulants of size
  $d$ over $\mathbb{F}_q$} Clearly the circulants are the winner in
this case. The circulants provide almost the same security as the
ElGamal over the finite field $\mathbb{F}_{q^d}$, but
multiplication in the circulants is twice as fast compared to the multiplication in the finite field $\mathbb{F}_{q^d}$. See Silverman~\cite{silverman1,silverman2} for more details.

To understand the difference, we need to understand the
standard field multiplication. A field $\mathbb{F}_{q^d}$ over
$\mathbb{F}_q$, an extension of degree $d$, is a commutative algebra of
dimension $d$ over $\mathbb{F}_q$. Let
$\alpha_0,\alpha_1,\ldots,\alpha_{d-1}$ be a basis of $\mathbb{F}_{q^d}$
over $\mathbb{F}_q$. Let
$A:=\left(a_0\alpha_0+a_1\alpha_1+\cdots+a_{d-1}\alpha_{d-1}\right)$,
$B:=\left(b_0\alpha_0+b_1\alpha_1+\cdots+b_{d-1}\alpha_{d-1}\right)$
and $$C:=A\cdot
B=\left(c_0\alpha_0+c_1\alpha_1+\cdots+c_{d-1}\alpha_{d-1}\right)$$ be elements of $\mathbb{F}_{q^d}$.

The objective of multiplication is to find $c_k$ for
$k=0,1,\ldots,(d-1)$. Now notice that, if
\[\alpha_i\alpha_j=\sum\limits_{k=0}^{d-1}t_{ij}^k\alpha_k,\] 
we can define a $d\times d$ matrix $T_k$ as $\{t_{ij}^k\}_{ij}$.
It follows that $c_k=AT_kB^{t}$. The number of nonzero entries in the
matrix $T_k$, which is constant over $k$, is called the \emph{complexity}
of the field multiplication~\cite[Chapter 5]{mullin}. The following
theorem is well known~\cite[Theorem 5.1]{mullin}:
\begin{theorem}
For any normal basis $N$ of $\mathbb{F}_{q^d}$ over $\mathbb{F}_q$, the
complexity of multiplication is at least $2d-1$.
\end{theorem}
Note that in an implementation of a field exponentiation, one must use
a normal basis to use the square and multiply algorithm.

In our case, circulants of size $d$ over a finite field
$\mathbb{F}_q$, the situation is much different. We need a normal
basis implementation for $\mathbb{F}_q$. However, to implement
multiplication of two circulants, i.e., multiplication in
$\mathcal{R}=\dfrac{\mathbb{F}_q[x]}{x^d-1}$ we can use the basis
$\left\{1,x,x^2,\ldots,x^{d-1}\right\}$. 

In a very similar way as
before, if $A:=a_0+a_1x+\ldots+a_{d-1}x^{d-1}$ and $B:=b_0+b_1x+\ldots
b_{d-1}x^{d-1}$ then $C:=A\cdot B=c_0+c_1x+\ldots+c_{d-1}x^{d-1}$. Our job is
to compute $c_k$ for $k=0,1,\ldots,d-1$. It follows that 
\begin{equation}c_k=\sum\limits_{i=0}^{d-1}a_ib_j\;\;\text{where}\;\;
  i+j=k\mod d\;\;\text{and}\;\;0\leq i,j\leq d-1\end{equation}
It is now clear that the complexity of the multiplication is $d$.
Compare this to the best case situation for the \emph{optimal normal
  basis}~\cite[Chapter 5]{mullin}, in
which case it is $2d-1$. So multiplying circulants take about half the
time that of finite fields.

It is clear that the keysizes will be the same for both these cryptosystems.
\subsection{The elliptic curve ElGamal vs.~the circulants of size $d$}
In this case there is no clear
winner. On one hand, take the case of embedding degree. For most
elliptic curves the embedding degree is very large.
The embedding degree, that we refer to as the security advantage, for a
circulant is tied up with the size of the matrix. For a matrix of size
$d$, it is $d-1$. So with circulants, it is hard to get very large
embedding degree, without blowing up the size of the matrix. On the
other hand, a very large embedding degree is not always necessary.

On the other hand, in elliptic curves, the order of the group is about the same as
the size of the field. For 80-bit security, we must take the field to
be around $2^{160}$, to defend against any square-root algorithms. In the
case of circulants, the order of a circulant matrix can be 
large. This enables us to \textbf{use smaller field for the same security}. In
circulants, one can use the extended Euclid's algorithm to compute the
inverse.

So, as we said before, we are not in a position to declare a clear winner in
this case. However, if the size of the field is important in the
implementation, and a moderate embedding degree suffices for security,
then circulants are a \textbf{little ahead in the game}. We explain this by some examples in the next section.

It is clear that the keysize for circulant matrices will be larger than that of the elliptic curve cryptosystem, both satisfying the following:
\begin{description} 
\item[1] Security of $80$ bits or more from generic algorithms. 
\item[2] Security from index-calculus comparable to the field $\mathbb{F}_{2^{1000}}$, i.e., \textbf{index calculus security} of $1000$ bits.
\end{description} 
\section{An algorithm}
Recall that C$(d,q)$ is isomorphic to $\dfrac{\mathbb{F}_q[x]}{x-1}\times\dfrac{\mathbb{F}_q[x]}{\Phi(x)}$. We now describe an algorithm to find a circulant matrix satisfying the above five conditions.
\begin{algorithm}[Construct a circulant matrix satisfying five conditions]\label{algo}\hspace*{\fill}\\

Input $q,d$.
\begin{itemize}
\item construct $\mathbb{F}_q$.
\item $\tau(x)\leftarrow$ A primitive polynomial of degree $d-1$ over $\mathbb{F}_q$.
\item order $\leftarrow$ Order of the determinant of the companion matrix of $\tau(x)$.
\item Use Chinese remainder theorem to find $\psi(x)$ such that $\psi(x)=1\mod(x-1)$ and $\psi(x)=\tau(x)\mod\Phi(x)$.
\item $\psi(x)\leftarrow\psi(x)\mod(x^d-1)$.
\item $A\leftarrow$ The circulant matrix with the first row $\psi(x)$.
\item $A\leftarrow A^{order}$.
\end{itemize}
Output $A$. 
\end{algorithm}
Using Magma~\cite{magma} and Algorithm~\ref{algo}, we were able to compute several circulant matrices over many different fields of characteristic 2. We produce part of that data in Table~\ref{table1}. The row with $q$ is the size of the field extension and the row with $d$ is the size of the circulant matrix over that field extension.

To construct the table, we considered all possible field extensions of size $q$, where $q$ varies from $2^{40}$ to $2^{100}$. For each such extension, we took all the primes, $d$, from $11$ to $50$. We then checked and tabulated the ones for which $q$ is primitive modulo $d$. For every extension $q$ and for all primes $d$, satisfying the primitivity condition, Algorithm~\ref{algo} was used and the output matrix was checked for all the five conditions and moreover the order of the matrix $A$ was found to be at least $q^{d-3}$. So, if $q$ is primitive modulo $d$, our algorithm produces the desired matrix $A$, satisfying all five conditions. The computation was fast on a standard workstation. 
 
\begin{table}
\begin{center}
\begin{tabular}{|l||c|c|c|c|c|c|}
\hline
$q$ & $2^{41}$ & $2^{43}$ & $2^{47}$ & $2^{49}$ & $2^{53}$ & $2^{55}$\\ \hline
$d$ & $11,13,19,$ & $11,13,19,$ & $11,13,19,$ & $11,13,19,$ & $11,13,19,$ & $13,19,$ \\
    & $29,37$  & $29,37$  & $37$  & $37$  & $29,37$  & $29,37$\\ \hline\hline
$q$ & $2^{59}$ & $2^{61}$  & $2^{65}$ & $2^{67}$ & $2^{71}$ & $2^{73}$\\ \hline
$d$ & $11,13,19,$ & $11,13,19,$ & $13,19,$ & $11,13,19,$ & $11,13,19,$ & $11,13,19,$\\
    & $29,37$  & $29,37$ & $29,37$ & $29,37$ & $29,37$ & $29,37$\\ \hline\hline
$q$ & $2^{77}$ & $2^{79}$ & $2^{83}$ & $2^{85}$ & $2^{89}$ & $2^{95}$\\ \hline
$d$ & $11,13,19,$ & $11,13,19,$ & $11,13,19,$ & $11,13,19,$ & $11,13,19,$ & $13,19,29,$\\
    & $37$ & $29,37$ & $37$ & $29,37$ & $29,37$ & $37$\\
\hline
\end{tabular}
\caption{Fields from size $2^{40}$ to $2^{100}$ and matrices from size $11$ to $50$ that satisfy those five conditions.}
\label{table1}
\end{center}
\end{table}
So now it is clear, that there are a lot of choices for parameters for the ElGamal cryptosystem over circulant matrices. We describe our findings with some arbitrary examples. For more data see Table~\ref{table2}.

In the case, $\mathbf{q=2^{89}, d=13}$, we found the largest prime factor of the order of $A$ to be $$7993364465170792998716337691033251350895453313.$$ The base two logarithm of this prime is $152.5$. So even if we use the Pohlig-Hellman algorithm to reduce the discrete logarithm in $A$, to the discrete logarithm problem in the prime factors of the order of $A$, we still have the security very close to the $80$-bit security from generic attacks. The security against the index calculus is the same as in $\mathbb{F}_{2^{1068}}$.

In case of $\mathbf{q=2^{39}, d=29}$, the largest prime factor of $A$ was \[3194753987813988499397428643895659569.\] The logarithm base $2$ of which is about $120$. So from generic attack, the 
security is about $2^{60}$ or sixty bit security. From index calculus the security is the same as the security of a field of size $\mathbb{F}_{2^{1092}}$.

In the case of $\mathbf{q=2^{45}, d=29}$, the largest prime factor of the order of $A$ is $15169173997557864184867895400813639018421$ with more than 60 bit security. The security against the index calculus is equivalent to $\mathbb{F}_{2^{1260}}$.

In the case of $\mathbf{q=2^{97},d=11}$, the largest prime divisor of $A$ is 
\begin{eqnarray*}50996843392805314313033252108853668830963472293743769141-\\06957559915561,\end{eqnarray*} the logarithm base $2$ is $231$. Security from generic attacks is $115$ bits and from index calculus is equivalent to the field $\mathbb{F}_{2^{970}}$, i.e., $970$ bits security.

In the case of $\mathbf{q=2^{43},d=29}$, the largest prime factor of the order is \begin{eqnarray*}1597133026914484603924687622599912490649282490944114-\\1855981389550399714935349,\end{eqnarray*} the logarithm of that is $253$. So this has about $125$ bit security from the generic attacks and $1204$ bit security from index calculus attack.

In the case of $\mathbf{q=2^{29},d=37}$, the largest prime factor is 
\[328017025014102923449988663752960080886511412965881,\] with logarithm $167$, i.e., security of more than $80$ bits from generic attacks and $1044$ bits from index calculus.

Using GAP~\cite{gap4}, we created Table~\ref{table2}. In this table, all extensions $q$, $q$ from $2^{45}$ to $2^{90}$ and all primes from $10$ to $20$ are considered. For those extensions and primes, it was checked if $q$ is primitive mod $d$. If that was so, then the circulant matrix $A$ was constructed and both the generic and the index calculus security was tabulated. 
\begin{center}
\input{table}
\end{center}
\subsection{Complexity of exponentiation of a  circulant matrix of size $d$}
Let us assume, that the circulant matrix of size $d$ is $A$ and we are raising it to power $m$, i.e., compute $A^m$. We are using the square and multiply algorithm. We know that squaring of circulants is free, and multiplication of two circulant matrices of size $d$ takes about $d^2$ field multiplications. The number of multiplications in the exponentiation is the same as the number of ones in the binary expansion of $m$. It is expected that a finite random string of zeros and ones will have about the same number of zeros and ones. So the expected number of ones in the binary expansion of $m$ is $\frac{1}{2}\log_2{m}$. So the expected number of field multiplications required to compute $A^m$ is $\frac{d^2}{2}\log_2{m}$. 
\bibliography{confpaper}
\bibliographystyle{amsplain} 
\end{document}

%% file: table.tex
{\small 
\begin{table}
\begin{center}
\begin{tabular}{||c|c|c|c||}\hline
Size of the   & Size of the & Logarithm of     & Index-calculus\\
extension $q$ & matrix $d$  &the largest prime & security in bits\\ \hline
              & $11$        & $115$            & $470$\\ \cline{2-4}
$2^{47}$       & $13$        & $77$             & $564$\\ \cline{2-4} 
              & $19$        & $207$            & $846$\\ \hline 
              & $11$        & $157$            & $490$\\ \cline{2-4}
$2^{49}$       & $13$        & $83$             & $588$\\ \cline{2-4}
              & $19$        & $112$            & $882$\\ \hline
$2^{51}$       & $11$        & $92$             & $510$\\ \hline
              & $11$        & $129$            & $530$\\ \cline{2-4}
$2^{53}$       & $13$        & $92$             & $636$\\ \cline{2-4}
              & $19$        & $312$            & $954$\\ \hline
$2^{55}$       & $13$        & $80$             & $660$\\ \cline{2-4}
              & $19$        & $239$            & $990$\\ \hline
$2^{57}$       & $11$        & $123$            & $570$\\ \hline
              & $11$        & $232$            & $590$\\ \cline{2-4}
$2^{59}$       & $13$        & $91$             & $708$\\ \cline{2-4} 
              & $19$        & $262$            & $1062$\\ \hline
              & $11$        & $157$            & $610$\\ \cline{2-4}
$2^{61}$       & $13$        & $120$            & $732$\\ \cline{2-4}
              & $19$        & $294$            & $1098$\\ \hline
$2^{63}$       & $11$        & $123$            & $630$\\ \hline
$2^{65}$       & $13$        & $96$             & $780$\\ \cline{2-4}
              & $19$        & $131$            & $1170$\\ \hline
              & $11$        & $248$            & $670$\\ \cline{2-4}
$2^{67}$      & $13$         & $106$           & $804$\\ \cline{2-4}
             & $19$         & $274$            & $1206$\\ \hline
$2^{69}$      & $11$         & $176$            & $690$\\ \hline
             & $11$         & $242$             & $710$\\ \cline{2-4}
$2^{71}$      & $13$         & $111$           & $852$\\ \cline{2-4}
             & $19$         & $281$            & $1278$\\ \hline
             & $11$         & $184$            & $730$\\ \cline{2-4}
$2^{73}$      & $13$        & $103$            & $876$\\ \cline{2-4}
             & $19$         & $258$            & $1314$\\ \hline
             & $11$        & $184$            & $770$\\ \cline{2-4}
$2^{77}$      & $13$        & $121$            & $924$\\ \cline{2-4}
             & $19$        & $359$            & $1386$\\ \hline
             & $11$        & $279$            & $790$\\ \cline{2-4}
$2^{79}$      & $13$        & $140$            & $948$\\ \cline{2-4}
             & $19$        & $209$            & $1422$\\ \hline
$2^{81}$      & $11$        & $143$           & $810$\\ \hline
             & $11$        & $284$            & $830$\\ \cline{2-4}
$2^{83}$      & $13$        & $132$           & $996$\\ \cline{2-4}
             & $19$        & $443$            & $1494$\\ \hline  
$2^{85}$      & $13$        & $101$           & $1020$\\ \cline{2-4}
             & $19$        & $245$            & $1530$\\ \hline
$2^{87}$      & $11$        & $151$            & $870$\\ \hline 
             & $11$        & $227$           & $890$\\ \cline{2-4}
$2^{89}$      & $13$        & $152$            & $1068$\\ \cline{2-4} 
             & $19$        & $323$            & $1602$\\ \hline

\end{tabular} 
\caption{Security for $q$ from $2^{45}$ to $2^{90}$ and $d$ from $10$ to $20$}
\label{table2}
\end{center}
\end{table}}